\documentclass[a4paper,twocolumn,11pt,accepted=2020-02-15]{quantumarticle}
\pdfoutput=1
\usepackage[utf8]{inputenc}
\usepackage[english]{babel}
\usepackage[T1]{fontenc}
\usepackage{amsmath}
\usepackage{hyperref}

\usepackage{tikz}
\usepackage{lipsum}

\usepackage[percent]{overpic}
\usepackage[english]{babel}
\usepackage{amsmath,amssymb,amsthm,amscd}
\usepackage{dsfont,hyperref,commath,mathtools,thmtools}
\usepackage{mathrsfs}
\usepackage{tcolorbox}
\usepackage{tikz}
\usetikzlibrary{arrows,positioning,shapes,shadows}
\input{qcircuit}

\newcommand{\qc}[1]{\begin{aligned} \Qcircuit  @C=10pt @R=4pt
    {#1} \end{aligned}}

\definecolor{bluegray}{RGB}{230,230,255}
\definecolor{paleyellow}{RGB}{255,255,204}
\tikzstyle{block}=[draw=black, rectangle, align=center, fill=bluegray, drop shadow]
\tikzstyle{rblock}=[draw=black, shape=rectangle,rounded corners=1em,align=center,fill=paleyellow, drop shadow]

\DeclareMathOperator{\Tr}{Tr}

\DeclareMathOperator{\spn}{span}
\DeclareMathOperator{\conv}{conv}
\DeclareMathOperator{\rng}{rng}

\newcommand{\bb}[1]{\mathbb{#1}}

\newcommand{\R}{\bb{R}}

\newtheorem{thm}{Theorem}
\newtheorem{lmm}{Lemma}
\newtheorem{cor}{Corollary}

\theoremstyle{definition}
\newtheorem{dfn}{Definition}
\newtheorem*{example*}{Example}

\begin{document}

\title{Extension  of  the Alberti-Uhlmann  criterion  beyond
  qubit dichotomies}

\author{Michele Dall'Arno}

\email{michele.dallarno@aoni.waseda.jp}

\affiliation{Centre   for  Quantum   Technologies,  National
  University  of  Singapore,  3  Science  Drive  2,  117543,
  Singapore}

\affiliation{Faculty  of Education  and Integrated  Arts and
  Sciences,    Waseda    University,   1-6-1    Nishiwaseda,
  Shinjuku-ku, Tokyo 169-8050, Japan}

\orcid{0000-0001-7442-4832}

\author{Francesco Buscemi}

\email{buscemi@i.nagoya-u.ac.jp}

\affiliation{Graduate  School  of  Informatics,  Nagoya
  University, Chikusa-ku, 464-8601 Nagoya, Japan}

\orcid{0000-0001-9741-0628}

\author{Valerio Scarani}

\email{physv@nus.edu.sg}

\affiliation{Centre   for  Quantum   Technologies,  National
  University  of  Singapore,  3  Science  Drive  2,  117543,
  Singapore}

\affiliation{Department of  Physics, National  University of
  Singapore, 2 Science Drive 3, 117542, Singapore}

\orcid{0000-0001-5594-5616}

\maketitle

\begin{abstract}
  The Alberti-Uhlmann criterion states  that any given qubit
  dichotomy can  be transformed  into any other  given qubit
  dichotomy by a quantum channel  if and only if the testing
  region of the former dichotomy includes the testing region
  of  the   latter  dichotomy.   Here,  we   generalize  the
  Alberti-Uhlmann criterion to the  case of arbitrary number
  of qubit  or qutrit states.   We also derive  an analogous
  result for the  case of qubit or  qutrit measurements with
  arbitrary   number  of   elements.   We   demonstrate  the
  possibility  of applying  our criterion  in a  semi-device
  independent way.
\end{abstract}

\section{Introduction}

When  quantum  states are  looked  at  as resources,  it  is
natural to study which states  can be transformed into which
others  by means  of  an allowed  set  of operations.   This
question  has  been  rephrased in  many  ways:  entanglement
processing, thermal operations... In this paper, we consider
generalizations  of  the following  task:  given  a pair  of
quantum    states     $(\rho_0,    \rho_1)$,     called    a
\textit{dichotomy},   determine   which  other   dichotomies
$(\sigma_0,   \sigma_1)$  can   be  obtained   from  it   by
application of a completely positive trace preserving (CPTP)
map. The  simplicity of the  problem is only  apparent: very
few results are known  about this problem.  Before reviewing
them, and stating our contribution,  let us take a detour to
consider the analogous task in classical statistics.

A classical dichotomy is a pair of probability distributions
$(p_0,p_1)$.    It  appears   naturally   in  the   simplest
formulation of  hypothesis testing,  in which there  are two
inputs  (the  \textit{null}   and  the  \textit{alternative}
hypotheses)    and   two    outputs   (\textit{accept}    or
\textit{reject}).  In this case,  any test is represented by
a  point  in   the  dichotomy's  \textit{hypothesis  testing
  region},        defined        as        the        region
$\{(p_0,p_1)\}\subset\mathbb{R}^2$   where   $p_0$  is   the
probability of  correctly accepting the null  hypothesis and
$p_1$  is  the probability  of  wrongly  accepting the  null
hypothesis with  the given  test~\cite{Ren16}. Tests  can be
then designed, for instance, to maximize $p_0$ while keeping
$p_1$ under  a certain  threshold. In particular,  the wider
the testing region, the more ``testable'', that is, the more
``distinguishable'' the pair of hypotheses is.

That the  testing region  is all  that matters  when dealing
with pairs of  hypotheses is made particularly  clear by the
celebrated Blackwell's theorem for dichotomies~\cite{Bla53}:
given two dichotomies $(p_0,p_1)$ and $(q_0, q_1)$, possibly
on  different  sample  spaces,  there  exists  a  stochastic
transformation that  transforms $p_0$  into $q_0$  and $p_1$
into $q_1$  simultaneously (``statistical  sufficiency'') if
and only if the testing  region for $(p_0,p_1)$ contains the
testing region for $(q_0, q_1)$.  In other words, the former
dichotomy can  be deterministically processed into  (or, can
deterministically simulate) the latter.  In the special case
in  which $p_1  = q_1  = u$,  the uniform  distribution, the
ordering induced  by comparing the testing  region coincides
with the ubiquitous  \textit{majorization ordering}: indeed,
the Lorenz curve corresponding to a probability distribution
$p$  is  nothing but  the  boundary  of the  testing  region
corresponding to the  dichotomy $(p, u)$~\cite{Bla53, Tor70,
  Tor91, Ren16}.

Such a compact characterization is  not known in the quantum
case  that concerns  us~\cite{RKW11, Jen12,  Mat14a, Mat14b,
  BG17}:  quantum  statistical  sufficiency  is  in  general
expressed   in    terms   of    an   infinite    number   of
conditions~\cite{Bus12, Mat10a, Jen16}  that are, therefore,
very  difficult  to  check in  practice~\cite{MOA11}.   Some
results,   based  on~\cite{Mat10b},   are  known   when  the
conversion is relaxed  to be approximate~\cite{BST19, WW19},
but  the  problem  remains   hard  in  general.   A  notable
exception  is the  case in  which both  quantum dichotomies,
$(\rho_0, \rho_1)$ and $(\sigma_0, \sigma_1)$, only comprise
two-dimensional   (i.e.,   qubit)   states.   Then,   as   a
consequence of  a well-known result by  Alberti and Uhlmann,
there exists a CPTP map transforming $(\rho_0, \rho_1)$ into
$(\sigma_0, \sigma_1)$ if and only  if the testing region of
the   former   contains   the    testing   region   of   the
latter~\cite{AU80}.    This  is   the   perfect  analog   of
Blackwell's theorem;  but counterexamples are known  as soon
as $(\rho_0,\rho_1)$ is a qutrit dichotomy~\cite{Mat14a}.

In this paper,  building upon previous works of  some of the
present authors~\cite{DBBV17, Dal19,  DBBT18, BD18, DHBS19},
we derive  the following  results. First,  we show  that any
family   of  $n$   qubit   states  which   can  all   become
simultaneously  real under  a single  unitary transformation
can be transformed  into any other family of  $n$ qubit (or,
under some conditions,  qutrit) states by a CPTP  map if the
testing region of the former  includes the testing region of
the latter (the Alberti-Uhlmann case is recovered for $n=2$,
since any  pair of qubit  states can be  made simultaneously
real). Second,  we show that  an analogous result  holds for
qubit or qutrit measurements with $n$ elements which can all
become   simultaneously   real   under  a   single   unitary
transformation. Our results follow  as a natural consequence
of the Woronowicz decomposition~\cite{Wor76} of linear maps,
once  families of  states and  measurements are  regarded as
linear transformations.   We demonstrate the  possibility of
witnessing   statistical   sufficiency  in   a   semi-device
independent  way, that  is,  without any  assumption on  the
devices except their Hilbert space dimension.

The     paper    is     structured    as     follows.     In
Section~\ref{sec:results}  we present  our main  results. We
first  introduce  our   extensions  of  the  Alberti-Uhlmann
criterion,    first    for    families   of    states    (in
Section~\ref{sec:simulability_states})    and    then    for
measurements                                             (in
Section~\ref{sec:simulability_measurement}).      We    then
discuss  semi-device  independent  applications,  first  for
families              of             states              (in
Section~\ref{sec:sdi_simulability_states}),  and   then  for
measurements                                             (in
Section~\ref{sec:sdi_simulability_measurement}).          In
Section~\ref{sec:proofs},  we provide  technical proofs  for
our            results.            In            particular,
Sections~\ref{sec:proofs_states}
and~\ref{sec:proof_measurement}   prove   the   results   of
Sections~\ref{sec:simulability_states}
and~\ref{sec:simulability_measurement},   respectively.   We
conclude     by      summarizing     our      results     in
Section~\ref{sec:conclusion}.

\section{Main results}
\label{sec:results}

We  will  make  use   of  standard  definitions  in  quantum
information   theory~\cite{Wil17}.   A   quantum  state   is
represented  by  a  density  matrix,  that  is,  a  positive
semi-definite operator $\rho$ such  that $\Tr[\rho] = 1$.  A
quantum   measurement   is   represented   by   a   positive
operator-valued measure, that is, a  family $\{ \pi_a \}$ of
positive   semi-definite   operators    that   satisfy   the
completeness  condition  $\sum_a  \pi_a =  \openone$,  where
$\openone$ denotes the identity operator.

A  channel is  represented  by a  completely positive  trace
preserving map, that  is, a map $\mathcal{C}$  such that for
any   state  $\rho$   one   has  $\Tr[\mathcal{C}(\rho)]   =
\Tr[\rho]$ and $(\mathcal{I} \otimes \mathcal{C}) (\rho) \ge
0$.  In the Heisenberg picture,  a channel is represented by
a completely  positive unit preserving  map, that is,  a map
$\mathcal{C}$  such  that  for  any  $\pi  \ge  0$  one  has
$\mathcal{C}(\openone) = \openone$ and $(\mathcal{I} \otimes
\mathcal{C}) (\pi) \ge 0$.

\subsection{Simulability of families of states}
\label{sec:simulability_states}

We say that a family $\{  \rho_x \}$ of $m$ states simulates
another (possibly different dimensional) family $\{ \sigma_x
\}$ of $m$ states, in formula
\begin{align}
  \label{eq:states_sim}
  \{ \sigma_x \} \preceq \{ \rho_x \},
\end{align}
if and  only if  there exists  a channel  $\mathcal{C}$ such
that $\sigma_x = \mathcal{C} (\rho_x)$ for any $x$.

If    condition~\eqref{eq:states_sim}   is    verified,   it
immediately follows that
\begin{align}
  \label{eq:states_rng_incl}
  \mathcal{R}  \left(  \left\{   \sigma_x  \right\}  \right)
  \subseteq  \mathcal{R}  \left(   \left\{  \rho_x  \right\}
  \right),
\end{align}
where  $\mathcal{R}( \{  \rho_x  \} )$  denotes the  testing
region of family  $\{ \rho_x \}$, defined as the  set of all
vectors whose  $x$-th entry is the  probability $\Tr [\rho_x
  \pi]$ for any measurement element $\pi$, in formula
\begin{align*}
  & \mathcal{R} \left( \left\{ \rho_x \right\} \right) \\ :=
  & \left\{ \mathbf{q} \; \Big| \;  \exists \; 0 \le \pi \le
  \openone  \textrm{  s.t.   }  \mathbf{q}_x  =  \Tr  \left[
    \rho_x \pi \right]\; \forall x \right\}.
\end{align*}
In  other words,  for any  measurement element  $\tau$ there
exists  a measurement  element $\pi$  such that  $\Tr[\rho_x
  \pi] = \Tr[\sigma_x \tau]$ for any $x$.

Here,  we   derive  conditions   under  which   the  reverse
implication      is       also      true,       that      is
Eq.~\eqref{eq:states_rng_incl}                       implies
Eq.~\eqref{eq:states_sim}:
\begin{thm}
  \label{thm:simulability_states}
  For any  family $\{ \sigma_x  \}$ of qubit states  and any
  real  family $\{  \rho_x  \}$ of  qubit  states (that  is,
  states that  have only  real entries  in some  basis), the
  following are equivalent:
  \begin{itemize}
    \item $\{ \sigma_x \} \preceq \{ \rho_x \}$.
    \item  $\mathcal{R}   (  \{  \sigma_x  \}   )  \subseteq
      \mathcal{R} ( \{ \rho_x \} )$.
  \end{itemize}
  If  $\{   \rho_x  \}$   contains  the   identity  operator
  $\openone$ in its linear span, the statement holds even if
  $\{ \sigma_x \}$ is a family of qutrit states.
\end{thm}
The proof is given in Section \ref{sec:proofs_states}.

Notice that the assumption that the family $\{ \rho_x \}$ of
states is real cannot be relaxed.  As a counterexample, take
$\{  \sigma_x  \}$  and  $\{  \rho_x  \}$  to  be  symmetric
informationally complete (or tetrahedral) families of states
with  $\rho_0 =  \sigma_1$  and $\rho_1  = \sigma_0$,  while
$\rho_k = \sigma_k$ for $k =  2, 3$. A family $\{ \rho_x \}$
of four pure qubit states is tetrahedral if and only if $\Tr
[  \rho_x \rho_z  ]$ is  constant for  any $x  \neq z$.   It
immediately follows  that there  exists a  transposition map
$\mathcal{T}$  (with  respect  to   some  basis)  such  that
$\sigma_x = \mathcal{T}  (\rho_x)$ for any $x$.   Due to the
informational  completeness  of  $\{ \sigma_x  \}$  and  $\{
\rho_x \}$, map $\mathcal{T}$ is the only map such that this
is the case. However, map  $\mathcal{T}$ is not a channel as
it is not completely positive.

\subsection{Simulability of measurements}
\label{sec:simulability_measurement}

We  say  that  an  $n$-outcome  measurement  $\{  \pi_a  \}$
simulates   another    (possibly,   different   dimensional)
$n$-outcome measurement $\{ \tau_a \}$, in formula
\begin{align}
  \label{eq:meas_sim}
  \left\{ \tau_a \right\} \preceq \left\{ \pi_a \right\},
\end{align}
if and  only if  there exists  a channel  $\mathcal{C}$ such
that  $\tau_a =  \mathcal{C}^\dagger (\pi_a)$  for any  $a$,
where $\mathcal{C}^\dagger$ denotes channel $\mathcal{C}$ in
the Heisenberg picture.

If  condition~\eqref{eq:meas_sim}  is verified,  it  follows
immediately that
\begin{align}
  \label{eq:meas_rng_incl}
  \mathcal{R}   \left(  \left\{   \tau_a  \right\}   \right)
  \subseteq  \mathcal{R}   \left(  \left\{   \pi_a  \right\}
  \right),
\end{align}
where the range $\mathcal{R}( \{  \pi_a \} )$ of measurement
$\{  \pi_a \}$  is defined  as  the set  of all  probability
distributions $\Tr[\rho \pi_a ]$ on the outcomes $a$ for any
state $\rho$, in formula
\begin{align*}
  & \mathcal{R} \left( \left\{  \pi_a \right\} \right) \\ :=
  & \left\{  \mathbf{p} \; \Big|  \; \exists \; \rho  \ge 0,
  \Tr \rho =  1, \textrm{ s.t.  } \mathbf{p}_a  = \Tr \left[
    \rho \pi_a \right] \; \forall a \right\}.
\end{align*}
In other words, for any  state $\sigma$ there exists a state
$\rho$ such that $\Tr[\rho  \pi_a] = \Tr[\sigma \tau_a]$ for
any $a$.

Similarly to what we did before, we derive
conditions under which the reverse implication is also true,
that      is       Eq.~\eqref{eq:meas_rng_incl}      implies
Eq.~\eqref{eq:meas_sim}:
\begin{thm}
  \label{thm:simulability_measurements}
  For any qubit or qutrit measurement $\{ \tau_a \}$ and any
  real qubit measurement  $\{ \pi_a \}$ (that  is, one whose
  elements are  all real in  some basis), the  following are
  equivalent:
  \begin{itemize}
    \item $\{ \tau_a \} \preceq \{ \pi_a \}$.
    \item   $\mathcal{R}  (   \{  \tau_a   \}  )   \subseteq
      \mathcal{R} ( \{ \pi_a \} )$.
  \end{itemize}
\end{thm}
The       proof       is        given       in       Section
\ref{sec:proof_measurement}. As before,  the assumption that
measurement $\{\pi_a \}$ is real cannot be relaxed.

\subsection{Semi-device independent simulability  of families of
  states}
\label{sec:sdi_simulability_states}

Suppose  that a  black box  preparator with  $m$ buttons  is
given, and  let us  denote with  $\rho_x$ the  unknown state
prepared upon the pressure of button $x$. Consider the setup
where a  black box tester  with $n$ buttons is  connected to
the  black  box  preparator,  and let  us  denote  with  $\{
\pi_{0|y}, \pi_{1|y}  := \openone  - \pi_{0|y} \}$  the test
performed upon the pressure of button $y$. One has
\begin{align}
  \label{eq:circuit}
  \qc{\lstick{x  \in [0,  m-1]}  &  \prepareC{\rho_x} \cw  &
    \multimeasureD{2}{\pi_{a|y}} \\ & & \nghost{\pi_{a|y}} &
    \rstick{a \in [0, 1]} \cw \\ & \lstick{y \in [0, n-1]} &
    \cghost{\pi_{a|y}}}
\end{align}
For each $y$, by  running the experiment asymptotically many
times   one   collects   the  vectors   $\mathbf{q}_y$   and
$\mathbf{u} - \mathbf{q}_y$ ($\mathbf{u}$ denotes the vector
with unit entries) whose  $x$-th entry are the probabilities
$\Tr[  \rho_x  \pi_{0|y}]$  and  $\Tr[  \rho_x  \pi_{1|y}]$,
respectively, that is
\begin{align*}
  \left[  \mathbf{q}_y   \right]_x  :=  \Tr   \left[  \rho_x
    \pi_{0|y} \right].
\end{align*}

We call semi-device independent  simulability the problem of
characterizing the class of all  families of states that can
be  simulated by  the black  box $\{  \rho_x \}$,  for which
simulability  can be  certified based  on distributions  $\{
\mathbf{q}_y  \}$  and  $\{ \mathbf{u}  -  \mathbf{q}_y  \}$
without any characterisation of the tests $\{ \pi_{a|y} \}$,
under an assumption on the Hilbert space dimension.

Here,   we   will   address  the   semi-device   independent
simulability problem  under the promise that  $\{ \rho_x \}$
is  a family  of qubit  states.  In  this case,  the testing
region~\cite{Dal19,  DBBT18}  is  the  convex  hull  of  the
isolated  points  $0$  and  $\mathbf{u}$  with  a  (possibly
degenerate)    ellipsoid    centered   in    $\mathbf{u}/2$.
Conversely, for any (possibly degenerate) ellipsoid centered
in $\mathbf{u}/2$ and contained in the hypercube $[0, 1]^m$,
its convex  hull with  $0$ and  $\mathbf{u}$ is  the testing
region  of a  qubit family  of states.   In general,  such a
testing  region  identifies  the  family  of  states  up  to
unitaries and anti-unitaries.

We will further make the  restriction that the black box $\{
\rho_x \}$ has $m = 2$  buttons, that is, $\{ \rho_x\}$ is a
dichotomy. Notice  that any  qubit dichotomy  is necessarily
real.   Hence,   in  the  discussion  above   the  (possibly
degenerate)  ellipsoid   becomes  a   (possibly  degenerate)
ellipse.   Additionally,  since two  anti-unitarily  related
qubit  dichotomies  are   also  unitarily-related,  a  qubit
dichotomy is identified by its range up to unitaries only.

Due  to Theorem~\ref{thm:simulability_states},  we have  the
following result:
\begin{cor}
  \label{cor:dds_states}
  If the convex hull of points $0$ and $\mathbf{u}$ with any
  given ellipse  centered in  $\mathbf{u}/2$ is a  subset of
  $\conv( 0, \mathbf{u}, \{ \mathbf{q}_y \}, \{ \mathbf{u} -
  \mathbf{q}_y \})$ it  is the testing region  of some qubit
  dichotomy that can be simulated by $\{ \rho_x \}$.
\end{cor}
Notice   that,  on   the   one  hand,   the  hypothesis   of
Corollary~\ref{cor:dds_states} represents  only a sufficient
condition  for a  qubit  dichotomy to  be  simulable by  $\{
\rho_x  \}$.   On  the  other  hand,  for  any  other  qubit
dichotomy that can be simulated  by $\{ \rho_x \}$ (if any),
simulability   cannot   be   certified  in   a   semi-device
independent way unless further data is collected.

As an application, consider the  case when one of the states
prepared  by the  dichotomy  (say $\rho_1$)  is the  thermal
state at infinite temperature, that is $\rho_1 = \openone/2$
(for this example we are assuming more than just the Hilbert
space  dimension,  although  no  knowledge  of  $\rho_0$  is
assumed).   Consider the  problem of  finding the  dichotomy
with maximal free  energy among those that  can be simulated
by  $\{ \rho_0,  \openone/2 \}$  through a  Gibbs-preserving
channel (in this case, a unit-preserving channel).

In this case, it immediately follows that the free energy is
monotone in  the area  of the  range.  This  can be  seen as
follows.  First, notice that the free energy in this case is
equal to the neg-entropy $-S(\rho_0)$, since the free energy
is equal  to the relative entropy  $S(\rho_0 || \openone/2)$
and  by  definition  one   has  $S(\rho_0||\openone/2)  =  -
S(\rho_0)$.   In turn,  $S(\rho_0) =  h(\lambda_\pm)$, where
$h(\cdot)$ denotes the binary  entropy and $\lambda_\pm$ the
eigenvalues of $\rho_0$. By setting $\lambda_{\pm} = 1/2 \pm
a$, by explicit computation  it immediately follows that the
volume  of  the  range  of $\{  \rho_0,  \openone/2  \}$  is
proportional to $a$, hence the statement is proved.

Suppose one test is performed on the black box dichotomy and
the following probability vectors are observed:
\begin{align}
  \label{eq:probability}
  \mathbf{q}_0     =      \frac12     \begin{pmatrix}1     -
    \epsilon\\1\end{pmatrix},     \qquad    \mathbf{u}     -
    \mathbf{q}_0     =     \frac12    \begin{pmatrix}1     +
      \epsilon\\1\end{pmatrix}.
\end{align}
for  some value  of parameter  $0 \le  \epsilon \le  1$. The
situation is illustrated in Fig.~\ref{fig:states_prob}.
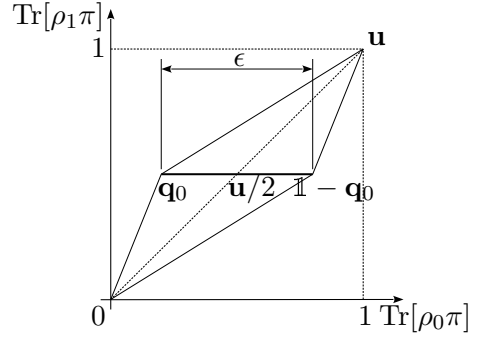
\begin{figure}[h!]
  \begin{center}
  \begin{overpic}[width=.49\columnwidth]{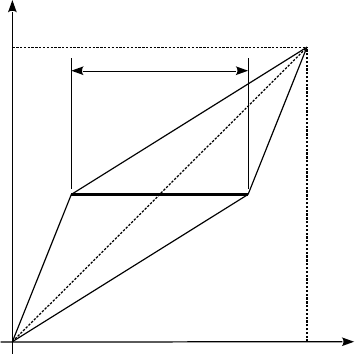}
    \put (-3, -5) {$0$}
    \put (-3, 84) {$1$}
    \put (85, -5) {$1$}
    \put (88, 88) {$\mathbf{u}$}
    \put (42, 38) {$\mathbf{u}/2$}
    \put (19, 38) {$\mathbf{q}_0$}
    \put (63, 38) {$\openone-\mathbf{q}_0$}
    \put (44, 81) {$\epsilon$}
    \put (92, -5) {$\Tr[\rho_0 \pi]$}
    \put (-29, 95) {$\Tr[\rho_1 \pi]$}
  \end{overpic}
  \end{center}
  \caption{Probability     vectors    $\mathbf{q}_0$     and
    $\mathbf{u}    -     \mathbf{q}_0$    as     given    by
    Eq.~\eqref{eq:probability} lie  at the  vertices of  a a
    line  segment  of  length  $\epsilon$  and  centered  in
    $\mathbf{u}/2$.  The maximally  committal testing region
    for  qubit   dichotomy  $\{  \sigma_0,   \openone/2  \}$
    enclosed   in    $\conv(0,   \mathbf{u},   \mathbf{q}_0,
    \mathbf{u}  -  \mathbf{q}_0)$   is  given  by  $\conv(0,
    \mathbf{u},  \mathbf{q}_0,  \mathbf{u} -  \mathbf{q}_0)$
    itself.}
  \label{fig:states_prob}
\end{figure}
We assume that the black box implements a qubit dichotomy, a
justified  assumption   since  the  probability   vector  in
Eq.~\eqref{eq:probability}  belongs,  for  example,  to  the
range of any  qubit dichotomy $\{ \phi,  \openone/2 \}$, for
any  pure state  $\phi$. It  is easy  to derive  the maximum
volume   range   enclosed    in   $\conv(   0,   \mathbf{u},
\mathbf{q}_0, \mathbf{u}  - \mathbf{q}_0  )$, and  to verify
using Ref.~\cite{DBBV17} that it  correspond to the range of
any        $\epsilon$-depolarized       dichotomy        $\{
\mathcal{D}_\epsilon (  \phi), \openone/2 \}$, for  any pure
state $\phi$.

\subsection{Semi-device     independent      simulability     of
  measurements}
\label{sec:sdi_simulability_measurement}

Suppose that  a black box  measurement with $n$  outcomes is
given,  and   let  us   denote  with  $\pi_a$   the  unknown
measurement element corresponding  to outcome $a$.  Consider
the setup where  a black box preparator with  $m$ buttons is
connected to  the black box  measurement, and let  us denote
with $\rho_x$  the unknown state prepared  upon the pressure
of button $x$. One has
\begin{align}
  \label{eq:meas_circuit}
  \qc{\lstick{x  \in [0,  m-1]}  &  \prepareC{\rho_x} \cw  &
    \measureD{\pi_a} & \rstick{a \in [0, n-1]} \cw}
\end{align}
For each $x$, by  running the experiment asymptotically many
times    one   collects    the   probability    distribution
$\mathbf{p}_x$ of outcome $a$, that is
\begin{align*}
  \left[ \mathbf{p}_x  \right]_a := \Tr \left[  \rho_x \pi_a
    \right].
\end{align*}

We call semi-device independent  simulability the problem of
characterizing  the class  of all  measurements that  can be
simulated  by  the  black  box  $\{  \pi_a  \}$,  for  which
simulability  can be  certified based  on distributions  $\{
\mathbf{p}_x \}$ without any  characterisation of the states
$\{ \rho_x  \}$, under  an assumption  on the  Hilbert space
dimension.

Here,   we   will   address  the   semi-device   independent
simulability problem under the promise that $\{ \pi_y \}$ is
a qubit measurement.  In  this case, the range~\cite{DBBV17,
  DBBT18} is a (possibly degenerate) ellipsoid.  Conversely,
any   (possibly   degenerate)   ellipsoid  subset   of   the
probability simplex is the range of a qubit measurement.  In
general,  such  a range  identifies  the  measurement up  to
unitaries and anti-unitaries.

We will further make the  restriction that the black box $\{
\pi_a  \}$   has  $n  =   3$  outcomes.   Notice   that  any
three-outcome qubit  measurement is necessarily real  due to
the completeness  condition. Hence, in the  discussion above
the  (possibly  degenerate)  ellipsoid becomes  a  (possibly
degenerate)  ellipse.    Additionally,  since  three-outcome
anti-unitarily related qubit measurements are also unitarily
related, a  three-outcome measurement  is identified  by its
range up to unitaries only.

Due to  Theorem~\ref{thm:simulability_measurements}, we have
the following result:
\begin{cor}
  \label{cor:dds_meas}
  Any ellipse subset of $\conv( \{ \mathbf{p}_x \} )$ is the
  range of some qubit  three-outcome measurement that can be
  simulated by $\{ \pi_a \}$.
\end{cor}
Notice   that,  on   the   one  hand,   the  hypothesis   of
Corollary~\ref{cor:dds_meas}  represents  only a  sufficient
condition  for  a  qubit  three-outcome  measurement  to  be
simulable  by $\{  \pi_a \}$.   On the  other hand,  for any
other qubit three-outcome measurement  that can be simulated
by $\{ \pi_a \}$ (if  any), simulability cannot be certified
in  a semi-device  independent  way unless  further data  is
collected.

As an  application, consider  the problem of  finding, among
the measurements that can be  simulated by the black box $\{
\pi_a \}$, the one with maximal simulation power, quantified
according to  Theorem~\ref{thm:simulability_measurements} by
the volume of its range. Suppose $m$ states are fed into the
black-box  measurement and  the following  distributions are
observed:
\begin{align}
  \label{eq:distributions}
  \mathbf{p}_x = \begin{pmatrix}
    2 - 2\cos \theta_x\\
    2 + cos \theta_x - \sqrt{3} \sin \theta_x\\
    2 + cos \theta_x + \sqrt{3} \sin \theta_x
  \end{pmatrix},
\end{align}
where $\theta_x := 2  \pi x / m$ and $x  \in [0, m-1]$. This
situation is depicted in Fig.~\ref{fig:distributions}.
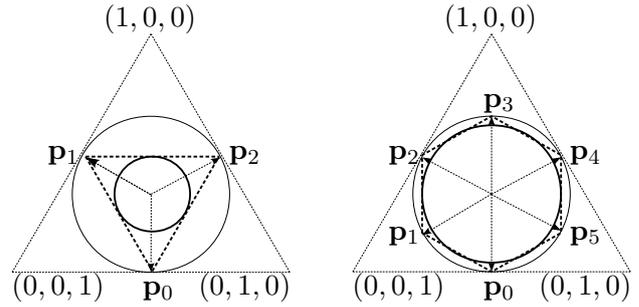
\begin{figure}[h!]
  \vspace{3mm}
  \begin{center}
    \begin{overpic}[width=.45\columnwidth]{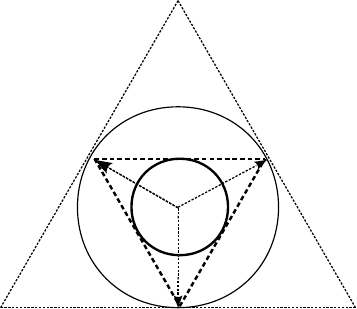}
      \put (33, 89) {$(1,0,0)$}
      \put (67, -8) {$(0,1,0)$}
      \put (0, -8) {$(0,0,1)$}
      \put (47, -8) {$\mathbf{p}_0$}
      \put (13, 41) {$\mathbf{p}_1$}
      \put (78, 41) {$\mathbf{p}_2$}
    \end{overpic}
    \hspace{.07\columnwidth}
    \begin{overpic}[width=.45\columnwidth]{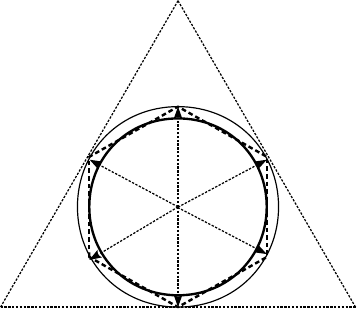}
      \put (33, 89) {$(1,0,0)$}
      \put (67, -8) {$(0,1,0)$}
      \put (0, -8) {$(0,0,1)$}
      \put (47, -8) {$\mathbf{p}_0$}
      \put (13, 12) {$\mathbf{p}_1$}
      \put (13, 41) {$\mathbf{p}_2$}
      \put (47, 60) {$\mathbf{p}_3$}
      \put (78, 41) {$\mathbf{p}_4$}
      \put (78, 12) {$\mathbf{p}_5$}
    \end{overpic}
  \end{center}
  \caption{In both left and  right figures, the outer dashed
    triangle   represents  the   simplex  of   three-outcome
    probability   distributions.    The  distributions   $\{
    \mathbf{p}_x  \}$ given  by Eq.~\eqref{eq:distributions}
    lie on the vertices of regular  polygons ($m = 3$ and $m
    =  6$ in  left  and right  figures, respectively).   The
    maximum   volume   ellipsoid   enclosed   in   $\conv(\{
    \mathbf{p}_x  \})$ is  the  inner circle,  which is  the
    range   of   a    $\epsilon$-depolarized   trine   qubit
    measurement   ($\epsilon   =   1/2$  and   $\epsilon   =
    \sqrt{3}/2$ in left and right figures, respectively).}
  \label{fig:distributions}
\end{figure}
We assume that the black box implements a qubit measurement,
a justified  assumption since  such distributions  belong to
the range of,  for example, a trine  qubit measurement, that
is, a  measurement whose elements  lie on the vertices  of a
regular triangle in the  Bloch sphere representation. It is
easy  to  derive  the  maximum  volume  ellipse~\cite{Joh48,
  Bal92, Tod16, BV04} enclosed in $\conv( \{ \mathbf{p}_x \}
)$,  and   to  verify   using  Ref.~\cite{DBBV17}   that  it
corresponds to the range of any $[\cos (\pi/m)]$-depolarized
trine measurement.

\section{Proofs of Theorems~\ref{thm:simulability_states} and~\ref{thm:simulability_measurements}}
\label{sec:proofs}

\setcounter{thm}{0}

The formalism of quantum information theory, used to present
our results  in Section~\ref{sec:results},  is not  the most
efficient to  prove such  statements.  Here, we  introduce a
more  efficient  formalism~\cite{DCP17}, that  provides  the
additional  benefit of  holding  for  any bilinear  physical
theory, not just quantum theory.

Each  system  is associated  with  a  dimension $\ell$,  and
states and  measurement elements are represented  by vectors
in $\R^\ell$.  Let us denote with $\mathbb{S}_\ell \subseteq
\R^\ell$ and $\mathbb{E}_\ell \subseteq  \R^\ell$ the set of
all  states  and  the   set  of  all  measurement  elements,
respectively.   The  probability  that  measurement  element
$\mathbf{e} \in  \mathbb{E}_\ell$ clicks  upon the  input of
state $\mathbf{s} \in \mathbb{S}_\ell$ is given by the inner
product  $\mathbf{s} \cdot  \mathbf{e}$.  Let  $\mathbf{u}_n
\in \R^n$  denote the  vector with unit  entries and  let us
choose    a   basis    in    which   $\mathbf{u}_\ell    \in
\mathbb{E}_\ell$ is  the measurement  element that  has unit
probability  of  click over  any  state.   For example,  for
quantum   systems  $\ell$   is  the   squared  Hilbert-space
dimension,   states   and   measurement  elements   can   be
represented  by  (generalized)  Pauli vectors,  their  inner
product  reduces to  the  Born  rule, and  $\mathbf{u}_\ell$
corresponds to the identity operator.

A  family  of  $n$  states   $\{  \mathbf{s}^k  \}$  can  be
conveniently represented by arranging the states as the rows
of  an   $n  \times  \ell$   matrix  $S$.   This   way,  the
corresponding linear  map $S  : \R^\ell  \to \R^n$  maps any
effect $\mathbf{e}$  into the  vector whose $k$-th  entry is
the   probability  $\mathbf{s}^k   \cdot  \mathbf{e}$.    It
immediately follows that, for any  family $S$ of states, one
has $S \mathbf{u}_\ell =  \mathbf{u}_n$. Analogously, an $n$
outcome  measurement  can  be  conveniently  represented  by
arranging its elements  $\{ \mathbf{e}^k \}$ as  the rows of
an $n \times \ell$ matrix  $M$.  This way, the corresponding
linear map $M \R^\ell \to  \R^n$ maps any state $\mathbf{s}$
into the probability distribution  whose $k$-th entry is the
probability $\mathbf{e}^k \cdot  \mathbf{s}$. It immediately
follows  that,  for  any   measurement  $M$,  one  has  $M^T
\mathbf{u}_n = \mathbf{u}_\ell$.

Finally,  we discuss  maps from  states to  states and  from
effects to effects.

\begin{dfn}[State morphism]
  A linear map $C :  \R^{\ell_0} \to \R^{\ell_1}$ is a state
  morphism  if   and  only  if  $C   \mathbb{S}_0  \subseteq
  \mathbb{S}_1$.
\end{dfn}

\begin{dfn}[Statistical morphism]
  A  linear  map $C  :  \R^{\ell_0}  \to \R^{\ell_1}$  is  a
  statistical  morphism  if  and  only  if  $C  \mathbb{E}_0
  \subseteq  \mathbb{E}_1$  and   $C  \mathbf{u}_{\ell_0}  =
  \mathbf{u}_{\ell_1}$.
\end{dfn}

In standard quantum  theory, a state morphism  is a positive
(not necessarily completely positive) trace-preserving (PTP)
map,   that  is,   a   transformation  of   states  in   the
Schr\"odinger picture.  Analogously,  a statistical morphism
is   a  positive   (not  necessarily   completely  positive)
unit-preserving  (PUP) map,  that  is,  a transformation  of
measurement  elements in  the  Heisenberg  picture. For  our
proofs, we do  not need an analogous  of complete positivity
for arbitrary bilinear theories.

\subsection{Simulability of families of states}
\label{sec:proofs_states}

Here we prove Theorem~\ref{thm:simulability_states}, that we
report here for the reader's convenience.

\begin{thm}
  For any  family $\{ \sigma_x  \}$ of qubit states  and any
  real family $\{ \rho_x \}$  of qubit states, the following
  are equivalent:
  \begin{itemize}
    \item $\{ \sigma_x \} \preceq \{ \rho_x \}$.
    \item  $\mathcal{R}   (  \{  \sigma_x  \}   )  \subseteq
      \mathcal{R} ( \{ \rho_x \} )$.
  \end{itemize}
  If  $\{   \rho_x  \}$   contains  the   identity  operator
  $\openone$ in its linear span, the statement holds even if
  $\{ \sigma_x \}$ is a family of qutrit states.
\end{thm}

By adopting  the formalism  of bilinear theories,  we denote
with $S_0$ and $S_1$ the families of states corresponding to
$\{ \sigma_x \}$ and $\{  \rho_x \}$, respectively. To prove
the theorem, we need to distinguish two cases. First, let us
consider the  case when  the linear span  of $\{  \rho_x \}$
contains the  identity operator $\openone$, that  is, $S_1^+
S_1   \mathbf{u}_{\ell_1}   =  \mathbf{u}_{\ell_1}$,   where
$(\cdot)^+$ denotes the Moore-Penrose pseudoinverse.  In the
following Lemma  we show  that, under the  hypothesis $S_1^+
S_1  \mathbb{E}_1 \subseteq  \mathbb{E}_1$, range  inclusion
between  two  families  of   states  is  equivalent  to  the
existence of a statistical morphism between them.

\begin{lmm}
  \label{lmm:states}
  For any  families of states  $S_0 : \R^{\ell_0}  \to \R^n$
  and  $S_1 :  \R^{\ell_1} \to  \R^n$ such  that $S_1^+  S_1
  \mathbb{E}_1  \subseteq   \mathbb{E}_1$  and   $S_1^+  S_1
  \mathbf{u}_{\ell_1} =  \mathbf{u}_{\ell_1}$, the following
  are equivalent:
  \begin{enumerate}
  \item  \label{item:states_majorization} $S_0  \mathbb{E}_0
    \subseteq S_1 \mathbb{E}_1$,
  \item  \label{item:states_simulability}   there  exists  a
    statistical morphism  $C : \R^{\ell_0}  \to \R^{\ell_1}$
    such that $S_0 = S_1 C$.
  \end{enumerate}
\end{lmm}

\begin{proof}
  Implication   $\ref{item:states_majorization}   \Leftarrow
  \ref{item:states_simulability}$ is trivial.

  Implication   $\ref{item:states_majorization}  \Rightarrow
  \ref{item:states_simulability}$    can    be   shown    as
  follows. Let
  \begin{align}
    \label{eq:states_map}
    C := S_1^+ S_0.
  \end{align}

  Let us first show that  map $C$ is a statistical morphism.
  One has
  \begin{align*}
    C \mathbb{E}_0 = S_1^+  S_0 \mathbb{E}_0 \subseteq S_1^+
    S_1 \mathbb{E}_1 \subseteq \mathbb{E}_1,
  \end{align*}
  where the equality  follows from Eq.~\eqref{eq:states_map}
  and  the  inclusions  follow   from  the  hypothesis  $S_0
  \mathbb{E}_0  \subseteq S_1  \mathbb{E}_1$ and  $S_1^+ S_1
  \mathbb{E}_1   \subseteq    \mathbb{E}_1$,   respectively.
  Moreover,
  \begin{align*}
    C \mathbf{u}_{\ell_0} =  S_1^+ S_0 \mathbf{u}_{\ell_0} =
    S_1^+ \mathbf{u}_n,
  \end{align*}
  where the equalities follow from Eq.~\eqref{eq:states_map}
  and  from   the  hypothesis  $S_0   \mathbf{u}_{\ell_0}  =
  \mathbf{u}_n$,  respectively.   Since by  hypothesis  $S_1
  \mathbf{u}_{\ell_1} =  \mathbf{u}_n$, one also  has $S_1^+
  S_1 \mathbf{u}_{\ell_1} = S_1^+ \mathbf{u}_n$, and hence
  \begin{align*}
    C \mathbf{u}_{\ell_0} =  S_1^+ S_1 \mathbf{u}_{\ell_1} =
    \mathbf{u}_{\ell_1},
  \end{align*}
  where the  second inequality follows by  hypothesis. Hence
  map $C$ is a statistical morphism.

  Let  us now  show  that  $S_0 =  S_1  C$.  By  multiplying
  Eq.~\eqref{eq:states_map} from the left by $S_1$ one has
  \begin{align*}
    S_1 C = S_1 S_1^+ S_0.
  \end{align*}
  Since   $\spn  \mathbb{E}_0   =  \R^{\ell_0}$   and  $\spn
  \mathbb{E}_1  =   \R^{\ell_1}$,  from   $S_0  \mathbb{E}_0
  \subseteq S_1  \mathbb{E}_1$ one  has $\rng  S_0 \subseteq
  \rng S_1$.   Since $S_1 S_1^+$  is the projector  on $\rng
  S_1$, one has $S_0 = S_1 C$.
\end{proof}

It  is   easy  to  see   that  the  hypothesis   $S_1^+  S_1
\mathbb{E}_1        \subseteq        \mathbb{E}_1$        in
Lemma~\ref{lmm:states} is satisfied for  any family $S_1$ of
states   if   $\mathbb{E}_1$  is   a   $\ell_1$--dimensional
(hyper)--cone         with         $(\ell_1-1)$--dimensional
(hyper)--spherical      section     with      axis     along
$\mathbf{u}_{\ell_1}$, as is the  case for the qubit system,
where $\ell_1 = 4$.  In  this case, by additionally assuming
that the  family $\{ \rho_x \}$  of states is real,  that is
$S_1 T =  S_1$ where $T$ denotes the  transposition map with
respect  to some  basis, the  following Lemma  completes the
first        part        of       the        proof        of
Theorem~\ref{thm:simulability_states}.

\begin{lmm}
  \label{lmm:qubit_states}
  For  any qubit  or qutrit  family $S_0  : \R^{\ell_0}  \to
  \R^n$ of  states, with $\ell_0 =  4$ or $\ell_0 =  9$, and
  any qubit family $S_1 : \R^4 \to \R^n$ of states such that
  $S_1^+ S_1 \mathbf{u}_{\ell_1}  = \mathbf{u}_{\ell_1}$ and
  $S_1 T = S_1$, the following are equivalent:
  \begin{enumerate}
  \item      \label{item:qubit_states_majorization}     $S_0
    \mathbb{E}_0 \subseteq S_1 \mathbb{E}_1$,
  \item \label{item:qubit_states_simulability}  there exists
    CPTP map $C : \R^{\ell_0} \to \R^4$ such that $S_0 = S_1
    C$.
  \end{enumerate}
\end{lmm}

\begin{proof}
  Implication          $\ref{item:qubit_states_majorization}
  \Leftarrow     \ref{item:qubit_states_simulability}$    is
  trivial.

  Implication          $\ref{item:qubit_states_majorization}
  \Rightarrow  \ref{item:qubit_states_simulability}$ can  be
  shown  as follows.   Due to  Lemma~\ref{lmm:states}, there
  exists a statistical morphism  $C': \R^{\ell_0} \to \R^4$,
  hence a PUP  map, such that $S_0 = S_1  C'$.  Let us prove
  that there  exists a CPUP  map $C : \R^{\ell_0}  \to \R^4$
  such that $S_0 = S_1 C$.

  For   any   PUP  map   $C'   :   \R^{\ell_0}  \to   \R^4$,
  Woronowicz~\cite{Wor76} proved  that there exist $0  \le p
  \le 1$ and CPUP maps $C_0 : \R^{\ell_0} \to \R^4$ and $C_1
  : \R^{\ell_0} \to \R^4$ such that
  \begin{align}
    \label{eq:states_decomposition}
    C' = p C_0 + \left( 1 - p \right) T C_1.
  \end{align}
  One has
  \begin{align*}
    S_1 C' = & S_1 \left[ p C_0 + \left( 1 - p \right) T C_1
      \right] \\ = & S_1 \left[ p C_0 + \left( 1 - p \right)
      C_1 \right],
  \end{align*}
  where        the       equalities        follow       from
  Eq.~\eqref{eq:states_decomposition}    and     from    the
  hypothesis $S_1 T =  S_1$, respectively.  Since the convex
  combination of CPUP maps is CPUP, map $C := p C_0 + \left(
  1 - p \right) C_1$ is CPUP.
\end{proof}

Second, let us consider the case when the linear span of $\{
\rho_x   \}$  does   not  contain   the  identity   operator
$\openone$.  Since by linearity  it immediately follows that
any linear dependency  in the states $\{ \rho_x  \}$ must be
present also  in the  states $\{  \sigma_x \}$,  as formally
show in Lemma~\ref{lmm:dependencies},  the remaining part of
the   proof  directly   follows  from   the  Alberti-Uhlmann
criterion,        as         formally        shown        in
Lemma~\ref{lmm:alberti_ulhmann}.

\begin{lmm}
  \label{lmm:dependencies}
  For any families  $S_0 : \R^{\ell_0} \to \R^n$  and $S_1 :
  \R^{\ell_1}   \to  \R^n$   of   states   such  that   $S_0
  \mathbb{E}_0 \subseteq S_1 \mathbb{E}_1$,  if for some $k$
  there exists $\{ \lambda_i \in \R \}$ such that
  \begin{align*}
    s_1^k = \sum_{i \neq k} \lambda_i s_1^i,
  \end{align*}
  then one has
  \begin{align*}
    s_0^k = \sum_{i \neq k} \lambda_i s_0^i.
  \end{align*}
\end{lmm}

\begin{proof}
  By hypothesis, for any $e_0 \in \mathbb{E}_0$ there exists
  $e_1 \in \mathbb{E}_1$ such that
  \begin{align*}
    s_0^k \cdot e_0 = s_1^k \cdot e_1.
  \end{align*}
  Hence, for  any set  $\{ e_0^j \}  \subseteq \mathbb{E}_0$
  one has
  \begin{align*}
    s_0^k \cdot e_0^j = s_1^k  \cdot e_1^j = \sum_{i \neq k}
    \lambda_i s_1^i \cdot e_1^j  = \sum_{i \neq k} \lambda_i
    s_0^i \cdot e_0^j.
  \end{align*}
  Since $\spn \mathbb{E}_0 = \R^{\ell_0}$, it is possible to
  take  set  $\{  e_0^j  \in  \mathbb{E}_0  \}$  a  spanning
  set. Hence the thesis.
\end{proof}

\begin{lmm}
  \label{lmm:alberti_ulhmann}
  For any  qubit families $S_0 :  \R^4 \to \R^n$ and  $S_1 :
  \R^4   \to  \R^n$   of   states  such   that  $S_1^+   S_1
  \mathbf{u}_{\ell_1} \neq \mathbf{u}_{\ell_1}$ and $S_1 T =
  S_1$ where $T$ denotes  the transposition map with respect
  to some basis, the following are equivalent:
  \begin{enumerate}
  \item    \label{item:alberti_ulhmann_majorization}    $S_0
    \mathbb{E}_0 \subseteq S_1 \mathbb{E}_1$,
  \item    \label{item:alberti_ulhmann_simulability}   there
    exists CPTP map $C : \R^4 \to \R^4$ such that $S_0 = S_1
    C$.
  \end{enumerate}
\end{lmm}

\begin{proof}
  Implication       $\ref{item:alberti_ulhmann_majorization}
  \Leftarrow   \ref{item:alberti_ulhmann_simulability}$   is
  trivial.

  Implication       $\ref{item:alberti_ulhmann_majorization}
  \Rightarrow  \ref{item:alberti_ulhmann_simulability}$  can
  be  shown as  follows.

  By   the  hypothesis   $S_0  \mathbb{E}_0   \subseteq  S_1
  \mathbb{E}_1$, one has that for any $e_0 \in \mathbb{E}_0$
  there  exists a  $e_1 \in  \mathbb{E}_1$ such  that $s_0^k
  \cdot e_0  = s_1^k e_1$  for any $k$. Since  in particular
  this holds  for $k =  0, 1$,  by denoting with  $S_0'$ and
  $S_1'$  the families  of  states whose  rows are  $(s_0^0,
  s_0^1)$ and $(s_1^0, s_1^1)$, respectively, one has
  \begin{align*}
    S_0' \mathbb{E}_0 \subseteq S_1' \mathbb{E}_1.
  \end{align*}
  Hence, due to a result~\cite{BG17} by Buscemi and Gour, in
  turn based on a result~\cite{AU80} by Alberti and Uhlmann,
  there exists  a CPTP  map $C  : \R^4  \to \R^4$  such that
  $S_0' = S_1' C$.
  
  Due to the hypotheses  $S_1^+ S_1 \mathbf{u}_{\ell_1} \neq
  \mathbf{u}_{\ell_1}$ and $S_1 T =  S_1$, for any $k$ there
  exists $\lambda^k \in \R$ such that
  \begin{align*}
    s_1^k =  \lambda^k s_1^0  + \left(1 -  \lambda^k \right)
    s_1^1,
  \end{align*}
  that is, state  $s_1^k$ is a convex  combination of states
  $s_1^0$      and     $s_1^1$.       Hence,     due      to
  Lemma~\ref{lmm:dependencies}, one also has
  \begin{align*}
    s_0^k =  \lambda^k s_0^0  + \left(1 -  \lambda^k \right)
    s_0^1,
  \end{align*}
  that is, state  $s_0^k$ is a convex  combination of states
  $s_0^0$ and $s_0^1$.  Hence, by linearity, $S_0 = S_1 C$.
\end{proof}

This concludes the proof of Theorem~\ref{thm:simulability_states}.

\subsection{Simulability of measurements}
\label{sec:proof_measurement}

Here  we prove  Theorem~\ref{thm:simulability_measurements},
that we report here for the reader's convenience.

\begin{thm}
  For any qubit or qutrit measurement $\{ \tau_a \}$ and any
  real qubit  measurement $\{  \pi_a \}$, the  following are
  equivalent:
  \begin{itemize}
    \item $\{ \tau_a \} \preceq \{ \pi_a \}$.
    \item   $\mathcal{R}  (   \{  \tau_a   \}  )   \subseteq
      \mathcal{R} ( \{ \pi_a \} )$.
  \end{itemize}
\end{thm}

By adopting  the formalism  of bilinear theories,  we denote
with $M_0$  and $M_1$ the measurements  corresponding to $\{
\tau_a \}$ and $\{ \pi_a  \}$, respectively.  In contrast to
Theorem~\ref{thm:simulability_states},  for  whose proof  we
needed to distinguish two cases  based on whether the linear
span  of  $\{ \rho_x  \}$  contained  the identity  operator
$\openone$ or  not, due to completeness  for any measurement
$\{ \pi_a  \}$ one has  $\sum_a \pi_a = \openone$,  that is,
$M_1^+ M_1 \mathbf{u}_{\ell_1} = \mathbf{u}_{\ell_1}$, where
$(\cdot)^+$ denotes the Moore-Penrose pseudoinverse.  Hence,
the     proof     proceeds     along    the     lines     of
Lemmas~\ref{lmm:states} and~\ref{lmm:qubit_states}, that are
in  this   case  replaced   by  Lemmas~\ref{lmm:measurement}
and~\ref{lmm:qubit_measurement},   respectively.    In   the
following Lemma  we show  that, under the  hypothesis $M_1^+
M_1  \mathbb{S}_1 \subseteq  \mathbb{S}_1$, range  inclusion
between two measurements is equivalent to the existence of a
state morphism between them.

\begin{lmm}
  \label{lmm:measurement}
  For any measurements $M_0 : \R^{\ell_0} \to \R^n$ and $M_1
  : \R^{\ell_1} \to \R^n$  such that $M_1^+ M_1 \mathbb{S}_1
  \subseteq \mathbb{S}_1$, the following are equivalent:
  \begin{enumerate}
  \item      \label{item:measurement_majorization}      $M_0
    \mathbb{S}_0 \subseteq M_1 \mathbb{S}_1$,
  \item \label{item:measurement_simulability} there exists a
    state morphism  $C :  \R^{\ell_0} \to  \R^{\ell_1}$ such
    that $M_0 = M_1 C$.
  \end{enumerate}
\end{lmm}

\begin{proof}
  Implication           $\ref{item:measurement_majorization}
  \Leftarrow     \ref{item:measurement_simulability}$     is
  trivial.

  Implication           $\ref{item:measurement_majorization}
  \Rightarrow  \ref{item:measurement_simulability}$  can  be
  shown as follows. Let
  \begin{align}
    \label{eq:measurement_map}
    C := M_1^+ M_0.
  \end{align}

  Let us first  show that map $C$ is a  state morphism.  One
  has
  \begin{align*}
    C \mathbb{S}_0 = M_1^+  M_0 \mathbb{S}_0 \subseteq M_1^+
    M_1 \mathbb{S}_1 \subseteq \mathbb{S}_1,
  \end{align*}
  where        the        equality       follows        from
  Eq.~\eqref{eq:measurement_map}  and the  inclusions follow
  from  the  hypothesis   $M_0  \mathbb{S}_0  \subseteq  M_1
  \mathbb{S}_1$  and   $M_1^+  M_1   \mathbb{S}_1  \subseteq
  \mathbb{S}_1$,  respectively. Hence  map  $C$  is a  state
  morphism.

  Let  us now  show  that  $M_0 =  M_1  C$.  By  multiplying
  Eq.~\eqref{eq:measurement_map} from the  left by $M_1$ one
  has
  \begin{align*}
    M_1 C = M_1 M_1^+ M_0.
  \end{align*}
  Since   $\spn  \mathbb{S}_0   =  \R^{\ell_0}$   and  $\spn
  \mathbb{S}_1  =   \R^{\ell_1}$,  from   $M_0  \mathbb{S}_0
  \subseteq M_1  \mathbb{S}_1$ one  has $\rng  M_0 \subseteq
  \rng M_1$.   Since $M_1 M_1^+$  is the projector  on $\rng
  M_1$, one has $M_0 = M_1 C$.
\end{proof}

It  is   easy  to  see   that  the  hypothesis   $M_1^+  M_1
\mathbb{S}_1        \subseteq        \mathbb{S}_1$        in
Lemma~\ref{lmm:measurement} is satisfied for any measurement
$M_1$    if    and    only   if    $\mathbb{S}_1$    is    a
$(\ell-1)$--dimensional  (hyper)--sphere  with center  along
$\mathbf{u}_{\ell_1}$, as is the  case for the qubit system,
where  $\ell_1   =  4$.

To see  this, notice  that the  (hyper)--sphere is  the only
body for which  there exists a point (the  center) such that
any line through  the point is orthogonal to  the surface of
the body.  Hence, the (hyper)--sphere  is also the only body
for  which  the  projection  of the  body  on  any  subspace
containing such  a point is  a subset of the  body. Finally,
notice that  by multiplying condition $M_1^T  \mathbf{u}_n =
\mathbf{u}_{\ell_1}$ on  the left  by $M_1^+ M_1$  and using
the  elementary property  of pseudoinverse  that $M_1^+  M_1
M_1^T = M_1^T$, one immediately has
\begin{align*}
  M_1^+ M_1 \mathbf{u}_{\ell_1} = \mathbf{u}_{\ell_1},
\end{align*}
that is,  $M_1^+ M_1$  is the projector  on a  subspace that
contains the center of the (hyper)--sphere $\mathbb{S}_1$.

In this case, by  additionally assuming that the measurement
$\{  \pi_a \}$  is real,  that is  $M_1 T  = M_1$  where $T$
denotes the  transposition map  with respect to  some basis,
the    following    Lemma    completes    the    proof    of
Theorem~\ref{thm:simulability_measurements}.

\begin{lmm}
  \label{lmm:qubit_measurement}
  For any qubit or qutrit measurement $M_0 : \R^{\ell_0} \to
  \R^n$, with  $\ell_0 = 4$ or  $\ell_0 = 9$, and  any qubit
  measurement $M_1 : \R^4 \to \R^n$  such that $M_1 T = M_1$
  where $T$  denotes the  transposition map with  respect to
  some basis, the following are equivalent:
  \begin{enumerate}
  \item   \label{item:qubit_measurement_majorization}   $M_0
    \mathbb{S}_0 \subseteq M_1 \mathbb{S}_1$,
  \item   \label{item:qubit_measurement_simulability}  there
    exists CPTP  map $C  : \R^{\ell_0}  \to \R^4$  such that
    $M_0 = M_1 C$.
  \end{enumerate}
\end{lmm}

\begin{proof}
  Implication     $\ref{item:qubit_measurement_majorization}
  \Leftarrow  \ref{item:qubit_measurement_simulability}$  is
  trivial.

  Implication     $\ref{item:qubit_measurement_majorization}
  \Rightarrow \ref{item:qubit_measurement_simulability}$ can
  be shown as  follows.  Due to Lemma~\ref{lmm:measurement},
  there exists a state  morphism $C': \R^{\ell_0} \to \R^4$,
  hence a PTP  map, such that $M_0 = M_1  C'$.  Let us prove
  that there  exists a CPTP  map $C : \R^{\ell_0}  \to \R^4$
  such that $M_0 = M_1 C$.

  For  any  PTP  map  $C' :  \R^{\ell_0}  \to  \R^4$,  there
  exists~\cite{Wor76} $0 \le  p \le 1$ and CPTP  maps $C_0 :
  \R^{\ell_0} \to \R^4$ and $C_1: \R^{\ell_0} \to \R^4$ such
  that
  \begin{align}
    \label{eq:decomposition}
    C' = p C_0 + \left( 1 - p \right) T C_1.
  \end{align}
  One has
  \begin{align*}
    M_1 C' = & M_1 \left[ p C_0 + \left( 1 - p \right) T C_1
      \right] \\ = & M_1 \left[ p C_0 + \left( 1 - p \right)
      C_1 \right],
  \end{align*}
  where        the       equalities        follow       from
  Eq.~\eqref{eq:decomposition} and from  the hypothesis $M_1
  T =  M_1$, respectively.  Since the  convex combination of
  CPTP maps is CPTP, map $C :=  p C_0 + \left( 1 - p \right)
  C_1$ is CPTP.
\end{proof}

This         concludes          the         proof         of
Theorem~\ref{thm:simulability_measurements}.

\section{Conclusion}
\label{sec:conclusion}

In   this  work   we  addressed   the  problem   of  quantum
simulability, that  is, the  existence of a  quantum channel
transforming a given device into another.  We considered the
cases of families of $n$ qubit or qutrit states and of qubit
or qutrit measurements with $n$ elements, thus extending the
Alberti-Uhlmann criterion  for qubit dichotomies.   Based on
these results, we demonstrated the possibility of certifying
the simulability in a  semi-device independent way, that is,
without any  assumption of the devices  except their Hilbert
space dimension.

\section*{Acknowledgement}

M.D.  is  grateful to A.   Bisio, A.  Jen\v{c}ov\'a,  and K.
Matsumoto for insightful discussions. This work is supported
by  the MEXT  Quantum  Leap Flagship  Program (MEXT  Q-LEAP)
Grant No.  JPMXS0118067285; the National Research Foundation
and the Ministry of Education, Singapore, under the Research
Centres of  Excellence programme; the Japan  Society for the
Promotion of  Science (JSPS)  KAKENHI, Grant  No.  19H04066;
the program for FRIAS-Nagoya IAR Joint Project Group.

\end{document}